\documentclass[12pt]{article}
\usepackage{amsmath,amssymb,amsthm,bm,xcolor,cite,fullpage,graphicx}
\usepackage{subcaption}

\title{Quasi-pinning and entanglement
\\
in the lithium isoelectronic series}

\author{Carlos L. Benavides-Riveros$^{1,2,3}$, Jos\'e M.
Gracia-Bond\'ia$^{1,2}$ \\
and Michael Springborg$^3$
\\ \\
$^1$Departamento de F\'isica Te\'orica
\\
Universidad de Zaragoza, 50009 Zaragoza, Spain
\\ \\
$^2$Instituto de Biocomputaci\'on y F\'isica de Sistemas Complejos 
(BIFI)
\\
Universidad de Zaragoza, 50018 Zaragoza, Spain
\\ \\
$^3$Physikalische und Theoretische Chemie 
\\
Universit\"at des Saarlandes, 66123 Saarbr\"ucken, Germany \\ \\
}

\date{\today}

\DeclareMathOperator{\Tr}{Tr}       

\newcommand{\al}{\alpha}            
\newcommand{\Up}{\Upsilon}          
\newcommand{\io}{\iota}             
\newcommand{\dl}{\delta}            
\newcommand{\ga}{\gamma}            
\newcommand{\la}{\lambda}           
\newcommand{\om}{\omega}            
\newcommand{\vf}{\varphi}           
\newcommand{\vs}{\varsigma}         

\newcommand{\bra}[1]{\langle#1|}    
\newcommand{\D}{\mathcal{D}}        
\newcommand{\dn}{{\mathord{\downarrow}}} 
\newcommand{\half}{\tfrac{1}{2}}    
\renewcommand{\H}{\mathcal{H}}      
\newcommand{\ket}[1]{|#1\rangle}    
\newcommand{\ox}{\otimes}           

\newcommand{\up}{{\mathord{\uparrow}}} 
\newcommand{\word}[1]{\quad\mbox{#1}\quad} 
\newcommand{\x}{\times}             

\newcommand{\cu}{\wedge} 

\newcommand{\vecform}{\bm}              
\newcommand{\rr}{\vecform{r}}           
\newcommand{\xx}{\vecform{x}}           

\newtheorem{thm}{Theorem}               

\newcommand{\Djokovic}{\mbox{D\kern-.8em\raise.15ex\hbox{--}\kern.35em}okovi\'c}

\makeatletter
\def\section{\@startsection{section}{1}{\z@}{-3.5ex plus -1ex minus
 -.2ex}{2.3ex plus .2ex}{\large\bfseries}}
\def\subsection{\@startsection{subsection}{2}{\z@}{-3.25ex plus -1ex
 minus -.2ex}{1.5ex plus .2ex}{\normalsize\bfseries}}
\makeatother

\begin{document}

\maketitle

\begin{abstract}
The Pauli exclusion principle gives an upper bound of 1 on the natural
occupation numbers. Recently there has been an intriguing amount of
theoretical evidence that there is a plethora of additional
generalized Pauli restrictions or (in)equalities, of kinematic
nature, satisfied by these numbers~\cite{Alturulato}. Here for the
first time a numerical analysis of the nature of such constraints is
effected in real atoms. The inequalities are nearly saturated, or
\textit{quasi-pinned}. For rank-six and rank-seven approximations for
lithium, the deviation from saturation is smaller than the lowest occupancy
number. For a rank-eight approximation we find well-defined families
of saturation conditions.
\end{abstract}




\newpage

\section{Introduction}

The natural occupation numbers, arranged in the customary decreasing
order $\la_1\ge\la_2\ge\cdots$ fulfil $0\leq\la_i \leq 1$, for all~$i$
---thus allowing no more than one electron in each quantum state.
Forty years ago Borland and Dennis~\cite{losprecursores} observed for
the rank-six approximation of a three-electron system, whose state
space is here denoted $\wedge^3\H_6$, that the six occupation numbers
satisfy the additional constraints $\la_{r}+\la_{7-r}=1$, where
$r\in\{1, 2,3\}$, allowing exactly one electron in the natural
orbitals $r$ and $7-r$. Moreover,
\begin{align}
\la_4 \leq \la_5 + \la_6.
\label{eq:constraints}
\end{align}
The proofs are given in~\cite{seapareciomaria}. Thanks to outstanding
work by Klyachko and others in the last few years, actually solving
the pure state $N$-representability problem for the one-body reduced
density matrix~\cite{Alturulato}, the pattern of the occupation
numbers has received renewed attention. Large sets of inequalities for
the eigenvalues of this matrix, widely
generalizing~\eqref{eq:constraints}, have been established. We note
that, while the pure $N$-representability problem for the two-body
reduced density matrix remains unsolved, the ensemble
$N$-representability problem for this matrix is now solved
\cite{Mazziotti}.

A recent article~\cite{pinning} proposes to give an analytic study
of the Klyachko conditions by means of a toy model: a one-dimensional
system of \textit{three spinless} fermions confined to a harmonic
well, interacting with each other through Hooke-type forces. A series
formula for the occupation numbers in terms of the coupling was found.

The tantalizing suggestion in \cite{pinning} is that the inequalities
are \textit{nearly saturated} in the ground state [i.e., in equations
like Eq.~(\ref{eq:constraints}) the equality almost holds]: this is
the ``\textbf{quasi-pinning}'' phenomenon, which points to a deep hold
on the kinematics of the system. Schilling \textit{et
al}~\cite{pinning} state ``\ldots \textsl{it is likely extremely
challenging to use numerical methods to distinguish between genuinely
pinned and mere quasi-pinned states}''.

In the work we report here, we have taken up this challenge by
studying the ground state of lithium-like ions, starting from scratch
with an elementary configuration interaction (CI) method, up to a
rank-eight approximation (here, the rank equals the number of basis
functions in setting up the CI expansion). This procedure serves a
twofold purpose. First, we shall study whether the conclusions of
Schilling {\it et al.} \cite{pinning} are valid for realistic systems,
too. There now exists a profound measure of quantum entanglement for
three-fermion systems in rank six~\cite{magyares}. A second goal of
the present work is therefore to adapt this measure to our physical
spin-partitioned systems, contrasting the results with the information
on entanglement traditionally provided by the~$\la_i$.

In the present paper we shall present our analysis and results as follows.

Section~\ref{sec:rankfive} gives a simple introduction to the problem
at hand. We discuss in some detail the one- and two-body matrices in
the relatively trivial approximation of rank five to the lithium-like
ground states $\wedge^3\H_5$.

In Section~\ref{sec:LV} we broach the subject of entanglement for our
systems. This will allow us to discuss subsequently the information-theoretic
meaning of pinning and quasi-pinning.

Section \ref{sec:seriousranks1} deals with the first non-trivial
approximation to the three-electron system (of rank six). We use two
different basis sets, and the comparison of the results turns out to be 
very instructive.

Section~\ref{sec:seriousranks2} analyzes the more complicated
cases of rank seven and eight approximations.
Finally, Section~\ref{sec:conclusion} summarizes our conclusion.

We emphasize that the calculated energies merely are used to provide 
information on the quality of our approximations and, accordingly, on how 
accurate our calculated occupation numbers are, including our conclusions
with regard to the quasi-pinning hypothesis. Our goal is to grapple
with its impact on chemistry, and to investigate the negative
correlation between it and entanglement. Thus we refrain completely from
gaining extra accuracy of machine calculations at the price of losing
insight.

In two appendices we give some additional mathematical information.
The first discusses the ideas behind the Klyachko constraints in
ordinary quantum chemical language. The second gives the proof of an
estimate that we shall present in Section~\ref{sec:seriousranks2}.

Finally, throughout this work we use Hartree atomic units.

\section{The simplest case: a rank-five configuration for lithium-alikes}
\label{sec:rankfive}

Consider a system of~$N$ electrons and $M$~spin orbitals
$\{\vf_i(\xx)\}^{M}_{i=1}$, each being a product of a spatial orbital
and a spinor. We employ the standard quantum-chemical notation $\xx :=
(\rr,\vs)$ and use the notational convention: $\vf_i (\xx) :=
\phi_{i}(\rr) \,\vs$, with $\vs \in \{\up,\dn\}$. The number of
configurations~$N_c$ that can be constructed from $M$ spin orbitals
for $N$ electrons and $M-N$ holes is
$$
N_c = \binom{M}{N},
$$
which grows as a factorial with $M$. Here, we assume that we have
identified a set of \textit{basis functions}, largely under the
guidance of the physical or chemical intuition~\cite{Saarland}, that
provides an accurate description of the system of our interest. For
the $N$-electron wave function, we use wave functions made of
normalized Slater determinants,
$$
|\Psi\rangle = \sum_J C_J \, [\vf_{J(1)} \cdots \vf_{J(N)}].
$$
With the exterior algebra notation, this becomes 
$$
[\vf_1 \vf_2 \cdots \vf_N] =: \tfrac1{\sqrt{N!}} | \vf_1 \rangle \cu |
\vf_2 \rangle \cu \cdots \cu | \vf _N\rangle.
$$
In general, we assume that the $\vf_i$ have been
orthonormalized, although we occasionally relate them to non-orthogonal
orbitals by
\begin{equation}
\vf_i(\xx) = \sum^{L}_{j=1}R_{ij} \, \psi_j(\rr,\varsigma).
\label{eq:matrizr}
\end{equation}

We define the following energy integrals:
\begin{align}
\kappa_{mn} &:= \int \frac{\vf_{m} (\xx) \, \vf_{n} (\xx)}{|\rr|} \,
d\xx, \qquad \pi_{mn} := -\frac12 \int \vf_{m} (\xx) \, \nabla^{2}_{\rr}
\vf_{n} (\xx) \, d\xx
\notag \\ 
\iota_{mnop} &:= \int \frac{\vf_{m} (\xx_1) \, \vf_{n} (\xx_1) \,
\vf_{o} (\xx_2) \, \vf_{p} (\xx_2)}{|\rr_1 - \rr_2|} \,
d\xx_1 \, d\xx_2 \notag;  
\notag \\
K_{mn} &:= \int \frac{\psi_{m} (\xx) \, \psi_{n} (\xx)}{|\rr|} \, d\xx,
\qquad P_{mn} := -\frac12 \int \psi_{m} (\xx) \, \nabla_{\rr}^{2} \psi_{n}
(\xx) \, d\xx \notag
\\
\Up_{mnop} &:= \int \frac{\psi_{m} (\xx_1) \, \psi_{n} (\xx_1) \,
\psi_{o} (\xx_2) \, \psi_{p} (\xx_2)}{|\rr_1 - \rr_2|} \, d\xx_1 \,
d\xx_2.
\end{align}
{}From one set of integrals one can construct other sets by means of 
the relations $\kappa= (R \ox R)\,K$, $\pi=(R \ox R)\,P$ and $\io=(R
\ox R\ox R\ox R)\,\Up$, where $R$ is the transformation matrix
in~\eqref{eq:matrizr}.

\subsection{A simple starting configuration}

Given its low ionization potential ($\simeq0.198$ au), it is natural
to explore radial configurations of the open-shell lithium atom with a
single-determinant composition of (a) two restricted helium-like spin
orbitals --- in turn motivated by the classical analysis by Shull and
L\"owdin~\cite{SL2} of the natural orbitals for spin singlet states
of~$\mathrm{He}$ --- and (b) one hydrogen-like, in a suitably general
sense. Specifically, in such a single configuration we use the Kellner
\textit{Ansatz} for the helium-like functions,
$$
\psi_1(\al, \rr) = \sqrt{\frac{\al^3}{\pi}}\ e^{-\al r}.
$$
For the spinor of the hydrogen-like function we have arbitrarily chosen $\dn$.
For the spatial orbital, typical textbook calculations can be used for the
$s$-orbital in the $L$-shell:
$$
\psi^s_3(\ga,\rr) = \frac14\sqrt{\frac{\gamma^3}{2\pi}}\,L^{1}_{1}(\gamma r)
\,e^{-\gamma r/2} = \frac14\sqrt{\frac{\gamma^3}{2\pi}}\,(2 - \gamma r)
\,e^{-\gamma r/2}.
$$
Moreover, we consider also the following functions,
$$
\psi^p_3(\gamma,\rr) = \frac14\sqrt{\frac{\gamma^5}{6\pi}}r \,
e^{-\gamma r/2} \word{or} \psi^d_3(\gamma, \rr) =
\frac18\sqrt{\frac{\gamma^7}{45\pi}}r^2 \, e^{-\gamma r/2}.
$$
With these functions we obtain better results than with $\psi^s_3$;
see Table~\ref{table:litioCI0}. The better approximation among the
three, which includes $\phi^p_3$, leads to a total energy that equals
$99.19\%$ of the ``exact'' value. Comparing to the Hartree--Fock (HF)
energy given by the ``best'' Slater determinant, the error is less
than $0.2\%$ --- much more satisfactory than the Kellner approximation
for helium.

\begin{table}
\centering      
\begin{tabular}{c c c c}  
Conf & Energy (au) & $\al$ & $\ga$ \\ [0.5ex]  
\hline\hline           
``exact''                          & $-7.478060$ & $-$ & $-$ \\
$\mathrm{HF}$                      & $-7.432727$ & $-$ & $-$ \\
$[\psi_1\dn\psi_1\up\psi^s_3\dn]$  & $-7.393597$ & $2.679747$ & $1.868327$ \\
$[\psi_1\dn\psi_1\up\psi^d_3\dn]$  & $-7.416163$ & $2.691551$ & $1.892738$ \\
$[\psi_1\dn\psi_1\up\psi^p_3\dn]$  & $-7.417919$ & $2.686435$ & $1.274552$ \\
[1ex]   
\hline    
\end{tabular} 
\caption{The exact, HF and variational energy of $\mathrm{Li}$ in a
single-determinant configuration. Note the more substantial screening
of the outer electron by the inner ones when including $\phi^p_3$ in the basis.}
\label{table:litioCI0} 
\end{table}

For higher~$Z$ in the lithium series, the accuracy naturally improves,
although we shall not discuss this issue further here. Notice instead
that the $R$-matrix mentioned above is just a Gram--Schmidt
orthonormalization matrix, i.e.,
\begin{align*}
\begin{pmatrix}
\phi_1\up \\ \phi_1\dn \\ \phi_3\dn \end{pmatrix} &= R \begin{pmatrix}
\psi_1\up \\ \psi_1\dn \\ \psi_3\dn \end{pmatrix}, \word{where} R =
\begin{pmatrix} 1 & 0 & 0 \\ 0 & 1 & 0 \\ 0 & -\frac{\langle\psi_1|
\psi_3\rangle}{\sqrt{1 - |\langle\psi_1|\psi_3 \rangle|^2}} &
\frac1{\sqrt{1 - |\langle\psi_1|\psi_3 \rangle|^2}} \end{pmatrix}.
\end{align*}
In order to simplify the presentation, we shall not give below the
explicit forms of such matrices.

\subsection{The rank-five computation}

We obtain the rank-five approximation by using two helium-like
one-particle wave functions and one hydrogen-like. Still being guided
by \cite{SL2}, for the former we add the following function of the set
(orthonormal on the ordinary space):
$$
\dl_n(r) := D_n\sqrt{\frac{\al^3}{\pi}}L^2_{n-1}\bigl(2\al
r\bigr)e^{-\al r}; \quad n = 1,2,\ldots
$$
where $D^{-2}_n=\binom{n-1}{2}$, and the associated Laguerre
polynomials $L^\zeta_n$ are as defined in~\cite{reliablerussian}. We
have thus
$$
\dl_2(\al, \rr) := \sqrt{\frac{\al^3}{3\pi}}L^2_1\bigl(2\al
r\bigr)e^{-\al r}.
$$
We shall adopt the following notation for an orthonormalized basis set
of the restricted spin-orbital type:
\begin{align*}
|1\rangle:=\vf^p_3\dn, \; |2\rangle := \dl_1\dn, \; |3\rangle &:=
\dl_2\dn, \; |4\rangle:=\dl_1\up, \; |5\rangle:=\dl_2\up;
\\[2\jot]
\word{where}
\begin{pmatrix}
|1\rangle \\ |2\rangle \\ |3\rangle \\ |4\rangle \\ |5\rangle
\end{pmatrix} &= R \begin{pmatrix} \psi^p_3\dn \\ \dl_1 \dn \\
\dl_2\dn \\ \dl_1 \up \\ \dl_2\up \end{pmatrix}.
\end{align*}
With rank five, one has in principle $10=\tbinom53$ Slater
determinants. However, since the adopted Hamiltonian is independent of
the spin coordinates, only pure spin states are physically meaningful.
Obviously, there are only six determinants which are eigenvectors of
the operator $S_z$, namely,
\begin{align}
[124], \ [134], \ [125], \ [135], \ [234], \ [235].
\label{estadossz}
\end{align}
The total spin operator $S^2$ can be written as $S_-S_+ + S_z +
S^2_z$. It is clear that the states in \eqref{estadossz} are
eigenstates of the operator $S_z$ (and consequently of~$S_z^2$).
However, it is less clear whether they are eigenstates of
$S_-S_+$, too. It is easy to show that the wave function
\begin{align*}
|\Psi\rangle = A[124] + B[134] + C[125] + D[135] + E[234] + F[235]
\end{align*}
satisfies
$$
S_- S_+ |\Psi\rangle - |\Psi\rangle \propto (B-C)
\bigl([134] + [1'23] + [125]\bigr),
$$
where $|1'\rangle$ is a spin-up counterpart of $|1\rangle$.
Therefore $S_- S_+ |\Psi\rangle = |\Psi\rangle$ and
$S^2|\Psi\rangle=\frac{3}{4} |\Psi\rangle$ if and only if $B=C$.

Throughout the remaining parts of this paper, we have used a similar
approach to identify those spin-adapted combinations of Slater
determinants, that are eigenfunctions to $S^2$ and, accordingly, are
not ``spin-contaminated'' states.

Finally, the normalized wave function is written as
\begin{gather}
A [124] + B [125] + B[134] + D [135] + E [234] + F
[235],
\label{eq:wave5} \\
\word{with}  |A|^2 + 2|B|^2 + |D|^2 + |E|^2 + |F|^2 = 1.
\notag
\end{gather}
With rows and columns indexed by $\{1,\dots,5\}$, the corresponding
one-body density matrix is expressed by the matrix
{\scriptsize 
\renewcommand{\arraycolsep}{3pt} 
$$
\begin{pmatrix}
|A|^2 + 2|B|^2 + |D|^2  & BE^* + D F^* & -AE^* - BF^* & 0 & 0 \\
B^*E+D^*F & |A|^2 + |B|^2 + |E|^2+|F|^2 & AB^*+BD^* & 0 & 0 \\
-A^*E-B^*F & A^*B + B^*D & |B|^2 + |D|^2 + |E|^2 +|F|^2 & 0 & 0 \\
0 & 0 & 0 & |A|^2 + |B|^2 + |E|^2 & AB^* + BD^* + EF^* \\
0 & 0 & 0 & A^*B + B^*D + E^*F & |B|^2 + |D|^2 + |F|^2    
\end{pmatrix}\!.
$$
}%
In our case, 
$$
\rho_1(\xx_1,\xx_1') = 3
\int\Psi(\xx_1,\xx_2,\xx_3)\Psi^*(\xx'_1,\xx_2,\xx_3)\,d\xx_2\,d\xx_3.
$$
We can now conclude that only combinations of the form
$\ket{[abc]}\bra{[dbc]}$ will contribute (where the order of $a$, $b$,
and $c$, as well as of $d$, $b$, and $c$ can be changed when
simultaneously taking the appropriate signs into account). For
instance, $\ket{[124]}\bra{[125]}$ contributes with $AB^*$ to the $45$
matrix entry; $\ket{[134]}\bra{[234]}$ contributes with $-BE^*$ to the
$12$ entry, and so on. Note that the trace of
this matrix is equal to~3, as it should be.%
\footnote{This is a result of the global multiplication by the factor
equal to the number of electrons; as well as a division by 3!, coming
from the appropriate constant of the determinants; and the fact that
each multiplication of two Slater terms contributes~twice.}

We thus have $\la_1 + \la_2 + \la_{3'}= 2$ and $\la_{4'} + \la_5
= 1$ for the natural occupation numbers; the primes in the notation
are due to them being not yet in decreasing order.

By definition, in the basis of natural orbitals $\{\ket{\al_i}\}$, the
one-body density matrix is diagonal: $ \rho_1 = \sum^5_{i=1} \la_i\,
|\al_i\rangle \langle\al_i|$, already assuming that the occupation
numbers are arranged in decreasing order by interchanging $\la_{3'}$
with~$\la_{4'}$. Therefore, it is evident that a strong \textit{selection
rule} applies: we can rewrite the wave function for a three-electron
system in rank~five in terms of only \textit{two} configurations:
\begin{align}
|\Psi\rangle_{3,5} = a [\al_1\al_2\al_3] + d[\al_1\al_4\al_5]; \quad
|a|^2 + |d|^2 = 1, \; \la_2 = \la_3 = |a|^2 \geq |d|^2 = \la_4 =
\la_5.
\label{eq:rankfunf}
\end{align}

Through this example we have given a simple proof of a theorem stated
by Coleman \cite{themaster}. A more sophisticated proof is found in
\cite[Cor.~2]{chinochano}.

\subsection{Spectral analysis of the $n$-body and $n$-hole density
matrices on $\wedge^3\H_5$}

According to the Schmidt--Carlson--Keller duality\cite{themaster},
when applied to a three-electron system, the nonzero eigenvalues as
well as their multiplicities are the same for the one- and the
two-body matrices, i.e.,
\begin{align*}
\rho_2 = \sum^5_{i=1} \la_i \, |\om_i\rangle \langle\om_i|,
\word{where} c_j |\om_j\rangle := 3 \int \Psi(\xx_1,\xx_2,\xx_3)
\al_j^*(\xx_3)\ d\xx_3 \word{with} |c_j|^2 = \la_j.
\end{align*}
Thus, the eigenvectors of the two-body matrix associated to the wave
function~\eqref{eq:rankfunf} are given by
\begin{align*}
|\om_1\rangle = a [\al_2\al_3] + d [\al_4\al_5], \quad |\om_2\rangle =
[\al_1\al_3], \quad |\om_3\rangle = [\al_1\al_2], \quad |\om_4\rangle
= [\al_1\al_5], \quad |\om_5\rangle = [\al_1\al_4].
\end{align*}

For a system of $N$ particles and $M-N$ holes, the $n$-hole matrix is
hermitian and antisymmetric in each set of subindices, similar to what
is the case for the $n$-particle matrix. Additionally, it satisfies
the normalization conditions and sum rules,
$$
\Tr\eta_n = \binom{M-N}{n}; \qquad \int \eta_n \, d\xx_n =
\frac{M-N-n}n \, \eta_{n-1}.
$$
In the natural orbital basis, the one-hole matrix becomes
$$
\eta_1 = \sum^M_{i=1} (1-\la_i)\, |\al_i\rangle \langle\al_i|,
\word{with} \Tr \eta_1 = M-N,
$$
i.e., $M-N=5-3=2$ in our case; while the two-hole matrix is the
$Q$-matrix of lore, which for the lithium in the rank-five
approximation is
$$
\eta_2 = \sum^5_{i=1} \mu_i \, |h_i\rangle \langle h_i| = |h_1\rangle
\langle h_1|,
$$
where $\mu_i = 0$ if $ |\om_i\rangle$ is a single determinant and 
otherwise $\mu_i = \la_i$. Here, $|h_1\rangle := d\,[\al_2\al_3] + a\,
[\al_4\al_5]$. Note that $\eta_2$ is idempotent:
$$
\eta_2^2 = \bigl(|h_1\rangle \langle h_1|\bigr)^2 = \eta_2
\word{because} \langle h_1| h_1 \rangle = 1.
$$

\section{Preliminary discussion of entanglement in $\bigwedge^3\H_6$}
\label{sec:LV}

We consider two different approaches for obtaining six-rank
approximations for lithium-like ions. One is to work in a scheme of
fully restricted spin orbitals. Then, the sixth molecular orbital is chosen
as $\psi^p_3\,\up$. An alternative is to include $\dl_3\,\dn$ instead.

For convenience, we use the notation
\begin{equation}
\begin{pmatrix}
|1\rangle \\ |2\rangle \\ |3\rangle \\ |4\rangle \\ |5\rangle \\
|6\rangle \end{pmatrix} = R \begin{pmatrix} \dl_1\up \\ \dl_1\,\dn \\
\psi^p_3\,\dn \\ \dl_2\,\dn \\ \dl_2\up \\ \psi^p_3\,\up \end{pmatrix};
\qquad
\begin{pmatrix} |1\rangle \\ |2\rangle \\ |3\rangle \\ |4\rangle \\
|5\rangle \\ |6\rangle \end{pmatrix} = R \begin{pmatrix}
\dl_1\up \\ \dl_1\,\dn \\ \psi^p_3\,\dn \\ \dl_2\,\dn \\ \dl_2 \up \\
\dl_3\,\dn \end{pmatrix};
\label{eq:alabado}
\end{equation}
respectively, for the two cases.

Before discussing the two approaches in detail in the next section,
it is useful to first discuss the relations between 
chemistry and entanglement in each case, in the light
of~\cite{magyares} and of the quite recent analysis of universal
subspaces for fermionic systems~\cite{chinochano}. Without doubt, the
search for an entanglement measure for multipartite systems is among
the most important challenges facing quantum information theory
\cite{Enredados}. For three fermions, there have been some attempts to
generalize the Schmidt decomposition, widely used in bipartite
systems. Both~\cite{magyares} and~\cite{chinochano} focus on rank-six
descriptions, since these are the lowest non-trivial ones for
tripartite systems.

Neither of the choices made in those papers is well adapted to the
needs of chemistry, the first being too general, and the second too
restrictive. The measure of entanglement proposed in~\cite{magyares}
on the basis of cubic Jordan algebra theory does not take account of
spin-partitioning. A wave function $|\Phi\rangle$ belonging to the
abstract twenty-dimensional Hilbert space $\wedge^3\H_6$ is
considered. Given an ordered basis of $\wedge^3\H_6$ and
\begin{align}
|\Phi\rangle = \sum_{1\le i < j < k\le 6}  c_{ijk} [ijk],
\label{eq:vrana}
\end{align}
its amount of entanglement is analyzed in terms of the absolute value
of the expression
$$
\mathcal{T} : = 4 \bigl\{[\Tr(M_1 M_2) - \mu\nu]^2 - 4 \Tr(M_1^\#
M_2^\#) + 4\mu \det M_1 + 4\nu\det M_2\bigr\} \word{\!with\!} 0\leq
|\mathcal{T}| \leq 1,
$$
where the twenty amplitudes of \eqref{eq:vrana} are arranged in two
$3\x 3$ matrices and two scalars,
$$
M_1 := \begin{pmatrix} c_{156} & -c_{146} & c_{145} \\
c_{256} & -c_{246} & c_{245} \\
c_{356} & -c_{346} & c_{345} \end{pmatrix} ,
\quad M_2 :=
\begin{pmatrix} c_{234} & -c_{134} & c_{124} \\
c_{235} & -c_{135} & c_{125} \\
c_{236} & -c_{136} & c_{126}         
\end{pmatrix},
\quad  \mu :=  c_{123} \word{and}  \nu :=  c_{456}.   
$$
Here, $M^\#$ denotes the adjugate of a matrix~$M$, such that $MM^\#
=M^\#M=(\det M)I$. Under this measure, non-trivial tripartite
entanglement can take place in two inequivalent ways: those
with~$|\mathcal{T}| \neq0$ and those with~$|\mathcal{T}|=0$ ---
provided that then a pertinent dual wavefunction $\tilde\Phi$ is
different from zero. Although both cases exhibit genuine tripartite
entanglement (they are neither separable nor biseparable), there is no
unitary transformation relating the two types of states. The lowest
configuration of the energy with the basis set $\{\psi^p_3 \dn,
\psi^p_3 \up, \dl_1 \up, \dl_1 \dn, \dl_2 \up, \dl_2 \dn\}$ considered
in the first part of this chapter results in a $\mathcal{T}$-measure
of entanglement equal to~\textbf{\textit{zero}}. In contrast, the wave
function constructed from $\{\psi^p_3\dn,\dl_1 \up,\dl_1\dn,\dl_2
\up,\dl_2\dn,\dl_3\dn\}$ results in $\mathcal{T}$-entanglement equal
to~$2.57 \x 10^{-6}$ (admittedly small, due to \textit{quasi-pinning},
as we explain later); which in particular means that entanglement-wise
pinned states and unpinned ones are mutually disconnected.

On the other hand, the framework of the analysis in~\cite{chinochano}
is applicable for only the first of the two configurations mentioned
in~\eqref{eq:alabado}.

\section{Rank-six approximations}
\label{sec:seriousranks1}

\subsection{Choosing two configurations}

It is readily seen that for the first basis set in \eqref{eq:alabado},
out of $20=\binom63$ Slater determinants there are nine eigenfunctions
of~$S_z$ with eigenvalue~$\dn$,
\begin{equation}
[123], \; [124], \; [245], \; [345], \; [236], \; [346], \; [134], \;
[246], \; [235].
\label{eq:notation1}
\end{equation}
The first six Slater determinants are eigenvectors of~$S^2$, which
also is true for the combinations
$$
[134] + [246] \word{and} [235] - [134].
$$
Consider thus the following wavefunctions
\begin{align*}
A[123] + B\big([235] - [134]\big) + E[124] + F[245] + D[345] + G[236] +
H[346] + I\big([246] + [134]\big).
\end{align*}
The notation corresponds to that of~\eqref{eq:wave5}, with,
however, a numbering change. It is easy to see that the
corresponding one-body matrix has the following spin structure
$$
\rho_1 = \rho^\up_1 \oplus \rho^\dn_1
$$
whereby, with respective indices $\{1,5,6\}$ and $\{2,3,4\}$,
\begin{align}
\rho^\up_1 &= \begin{pmatrix} |A|^2 + |B|^2 + |E|^2 + |I|^2 & \star & \star
\\ \star & |B|^2 + |D|^2 + |F|^2 & \star \\ \star & \star & |G|^2 +
|H|^2 + |I|^2 \end{pmatrix}\!, \quad \Tr\rho^\up_1 = 1;
\notag \\[\jot]
\rho^\dn_1 &= \mbox{\scriptsize$\begin{pmatrix}
|A|^2 + |B|^2 + |E|^2 + |F|^2 + |G|^2 + |I|^2 & \star
& \star \\[2\jot]
\star & \hspace*{-4em} 
|A|^2 + 2|B|^2 + |D|^2 + |G|^2 + |H|^2 + |I|^2 \hspace*{-4em}
& \star \\[2\jot]
\star & \star & |B|^2 + |D|^2 + |E|^2 +|F|^2 + |H|^2 + 2|I|^2
\end{pmatrix}$},
\notag \\[\jot]
&\qquad \Tr\rho^\dn_1 = 2.
\label{eq:shouldbeclearbutmaybeisnot}
\end{align}

For the second basis system in \eqref{eq:alabado}, among the 20~Slater
determinants there are now twelve eigenfunctions of the operator $S^z$
with eigenvalue~$\dn$, namely,
$$
[123], \; [124], \; [245], \; [345], \; [134], \; [235], \; [146], \;
[256], \; [136], \; [356], \; [126], \; [456].
$$
Here, we shall not write explicitly the general wavefunction that can
be constructed from those and that does not contain any spin
contamination.

Table~\ref{table:litioCIenergies} presents the results for the energy
and screening parameters, with $6^a$ and $6^b$ respectively denoting
the first and second case in~\eqref{eq:alabado}. In the table we have also 
included the results for higher-rank approximations.%
\footnote{With our method it is necessary to reach rank~seven in order
to obtain part of the (radial) correlation energy. It is well known
that the best HF ground state for $\mathrm{Li}$ is given by an
unrestricted determinant.}

\begin{table}
\centering      
\begin{tabular}{c c c c}  
Rank  & Energy  & $\al$ & $\ga$ \\ [0.5ex] 
\hline\hline           
$3$   & $-7.417919$ & $2.686435$ & $1.274552$ \\
$5$   & $-7.431181$ & $2.711177$ & $1.304903$ \\
$6^a$ & $-7.431827$ & $2.674424$ & $1.319161$ \\
$6^b$ & $-7.431639$ & $2.712166$ & $1.323417$ \\
$7$   & $-7.445443$ & $2.772402$ & $1.336274$ \\
$8$   & $-7.454889$ & $2.767562$ & $1.331108$ \\
[1ex]   
\hline    
\end{tabular} 
\caption{Variational energy of $\mathrm{Li}$ in a CI picture for
different approximation rank.}
\label{table:litioCIenergies} 
\end{table}

Table~\ref{table:litioON678} gives the results for the natural
orbital occupancy numbers.

\begin{table}
 \centering      
{\footnotesize
\begin{tabular}{c c c c c c c c c }  
Rank & $\la_1$  & $ \la_2$ & $\la_3$ & $\la_4\x 10^{3}$ & $\la_5\x 10^{3}$ &
$\la_6 \x 10^{4}$   & $\la_7 \x 10^{5}$  & $\la_8 \x 10^{6}$
\\[0.5ex] 
\hline\hline           
$5$   &   $1$      & $0.998702$ & $0.998702$ & $1.297058$ & $1.297058$ & $-$      & $-$ & $-$ \\
$6^a$ & $0.999978$ & $0.998677$ & $0.998655$ & $1.344195$ & $1.322335$ & $0.2185$ & $-$ & $-$ \\
$6^b$ & $0.999977$ & $0.998715$ & $0.998715$ & $1.284753$ & $1.284182$ & $0.2203$ & $-$ & $-$ \\
$7$   & $0.999868$ & $0.998629$ & $0.998511$ & $1.416148$ & $1.364978$ & $1.2336$ & $8.5241$ & $-$ \\
$8$   & $0.999839$ & $0.998663$ & $0.998522$ & $1.409339$ & $1.337846$ & $1.3972$ & $8.6559$ & $1.7232$
\\[1ex]   
\hline    
\end{tabular} }
\caption{Occupation numbers from ranks five to eight for 
lithium wave functions.}
\label{table:litioON678} 
\end{table}

The (four) Klyachko inequalities for a three-electron system
in a rank-six configuration read
$$
\la_1 + \la_6 \le 1, \quad
\la_2 + \la_5 \le 1, \quad
\la_3 + \la_4 \le 1; \quad
0 \le \D :=  \la_5 + \la_6 - \la_4.
$$
However, one must have $\sum_{i=1}^6\la_i=3$. As a consequence of
this, the first inequalities become saturated (the Borland--Dennis
identities), and there is only one inequality left for further
examination. Note that we can formulate this as
\begin{equation}
\la_1 + \la_2 \le 1 + \la_3.
\label{eq:suerte-o-verdad}
\end{equation}

Before analyzing $\D$, which is the main subject in this
subsection, we emphasize that the Borland--Dennis identities are
fulfilled within our numerical accuracy. Also, they imply that in
the natural orbital basis every Slater determinant is composed of
three orbitals $[\al_i\al_j\al_k]$, each belonging to one of three
different sets, say
$$
\al_i \in \{\al_1,\al_6\}, \quad \al_j \in \{\al_2,\al_5\} \word{and}
\al_k \in \{\al_3,\al_4\};
$$ 
that is $\wedge^3\,\H_6$ splits with a section equal to
$\H_2^{\ox3}$ for a system of three fermions with spin.

\smallskip

\textbf{\textit{Quasi-pinning}} is the property of $\D$ being
extremely close to zero. Within our calculation~$6^b$, we find
\begin{align}
0 \le \D =  \la_5 + \la_6 - \la_4 = 2.1465 \x 10^{-5}.
\label{eq:pinning6++}
\end{align}
This value of $\D$ is slightly smaller than the lowest occupation
number, $\D/\la_6\approx 0.97$. $\D$~cannot exceed $\la_6$, because
otherwise $\la_5 >\la_4$. More remarkable is that for the restricted
determinant case $6^a$ one has $\D$ of order $10^{-12}$, i.e., 0
within numerical accuracy.

The inequality~\eqref{eq:suerte-o-verdad} together with the decreasing
ordering rule define a polytope~(Fig.~\ref{graf:polytope}) in the
space of the occupancy numbers.

\begin{figure}[ht] 
\centering
\includegraphics[width=8cm]{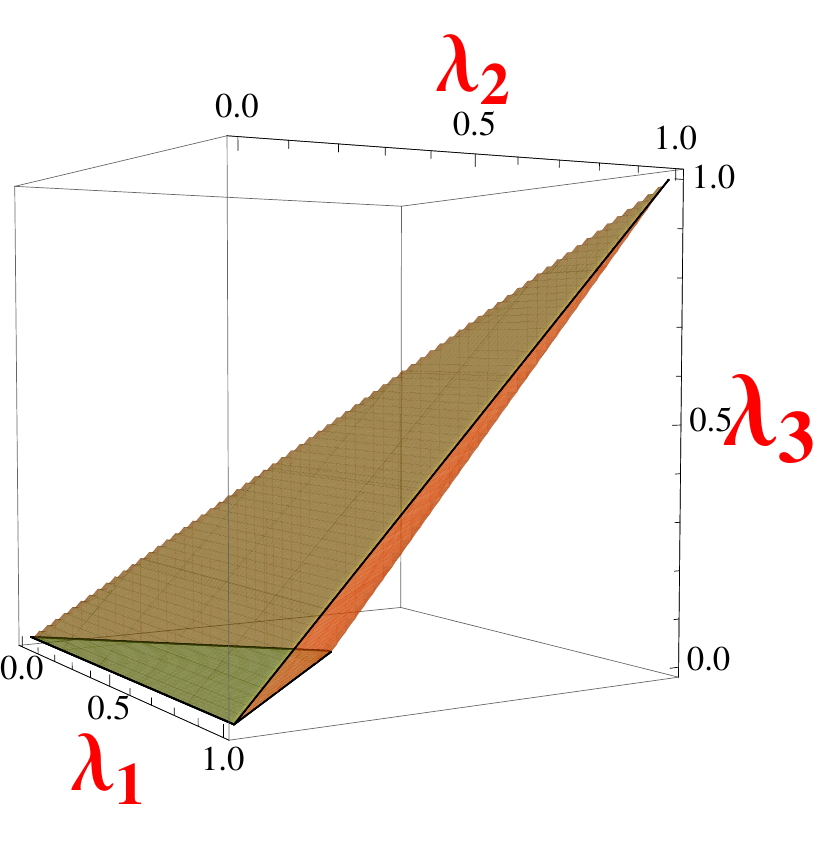}
\caption{(Color online) Polytope defined by the expression $\la_1 + \la_2 \le 1 +
\la_3$, subject to the condition $1 \ge \la_1 \ge \la_2 \ge \la_3 \ge
0$. The saturation condition $\la_1 + \la_2 = 1 + \la_3$ is satisfied
by the points on one of the faces of the polytope, the one with edges
$\la_2 = \la_3$ for $\la_1 = 1$ and $\la_1 = 1-\la_2$ for $\la_3 = 0$.
The single determinant state is placed at the corner $\la_1 = \la_2 = \la_3
= 1$ of the polytope. The physical ground states appear to be (close to)
saturated}
\label{graf:polytope}
\end{figure}

So far, a number of findings and conclusions can be emphasized:

\begin{itemize}

\item{} The energy computed via the restricted basis set $6^a$ is
(marginally) better than that obtained via~$6^b$.

\item{} Quasi-pinning is ``strict'' for $6^a$ --- in fact we do have
\textbf{pinning}--- and ``lax'' for~$6^b$. Indeed, equation
\eqref{eq:pinning6++} is still remarkable in absolute terms. But it
just means that if the system is close to a vertex, it is close to a
face.

\item{} Both states are truly entangled ---neither separable nor
biseparable. However, the $\mathcal T$-measure of entanglement
\textit{vanishes} for $6^a$, while $\mathcal T\ne0$ for $6^b$. Thus,
in some sense the latter is ``more entangled'' than the former. In
fact, referring to the original notation~\eqref{eq:notation1}, for the
case $6^a$ we have the expressions:
$$
M_1 = 
\begin{pmatrix}
0 & 0 & 0 \\
0 & -c_{246} & c_{245} \\
0 & -c_{346} & c_{345}         
\end{pmatrix} , \quad
M_2 = 
\begin{pmatrix}
0 & -c_{134} & c_{124} \\
c_{235} & 0 & 0 \\
c_{236} & 0 & 0       
\end{pmatrix},
\quad  \mu =  c_{123} \word{and}  \nu =  0,
$$
and hence
$$
\mathcal{T} = 4 \{[\Tr(M_1 M_2) - \mu\nu]^2 - 4 \Tr(M_1^\# M_2^\#) +
4\mu \det M_1 + 4\nu \det M_2\}  = 0.
$$

For the case $6^b$, again referring to the original
notation~\eqref{eq:notation1}, we deal with
$$
M_1 = \begin{pmatrix} 0 & -c_{146} & 0 \\ c_{256} & 0 & c_{245} \\
0 & 0 & c_{345} \end{pmatrix} , \quad M_2 = \begin{pmatrix} 0 &
-c_{134} & c_{124} \\ c_{235} & 0 & 0 \\ 0 & 0 & c_{126}
\end{pmatrix}, \quad \mu = c_{123} \word{and} \nu = c_{456},
$$
and hence
\begin{align*}
\mathcal{T} &= 4 \{[\Tr(M_1 M_2) - \mu\nu]^2 - 4 \Tr(M_1^\# M_2^\#) +
4\mu \det M_1 + 4\nu \det M_2\} 
\\
&= 4 \{(-c_{146} c_{235} - c_{134} c_{256} + c_{126} c_{345}
- c_{123} c_{456})^2 - 4(c_{134} c_{146} c_{235} c_{256}
- c_{126} c_{146} c_{235} c_{245} \\
&\qquad - c_{126} c_{134} c_{256} c_{345}) +
4 c_{123} c_{146} c_{256} c_{345} + 4 c_{456} c_{126} c_{134} c_{235}\}
= -2.5718 \x 10^{-6}.
\end{align*}

\item{} It is accordingly natural to conjecture, as already done
in~\cite{pinning}, that pinning leads to qualitative differences in
multipartite entanglement, and quasi-pinning correlates
\textit{negatively} with entanglement.

\item{} Computing entanglement by means of the standard Jaynes entropy
$-\sum_i\la_i\ln\la_i$ we obtain for the restricted configuration
$2.05 \x 10^{-2}$ and $1.99 \x 10^{-2}$ for the partially unrestricted
one. Admittedly, these two values are close but nevertheless it would
seem to \textit{contrarily} indicate that $6^b$ is ``less entangled''
than~$6^a$. In total, this suggests that there is a need to
identify genuine multipartite measures of entanglement. The recent
proposal~\cite{los4mosqueteros} looks enticing in this respect.

\item{} When the fourth inequality saturates ($\D=0$), a strong
selection rule like~\eqref{eq:rankfunf} applies, namely, the number of
Slater determinants reduces to three:
\begin{align}
|\Psi\rangle_{3,6} = a[\al_1\al_2\al_3] + b[\al_1\al_4\al_5] +
c[\al_2\al_4\al_6] .
\label{eq:sixpinned}
\end{align}
It should be clear that $\{\al_1,\al_2,\al_4\},\,\{\al_3,\al_5,
\al_6\}$ respectively span the spaces on which $\rho_1^\up,
\,\rho_1^\dn$ in~\eqref{eq:shouldbeclearbutmaybeisnot} act. The
natural occupation numbers for this wavefunction are of the form:
$$
\la_1 = |a|^2 + |b|^2, \quad \la_2 = |a|^2 + |c|^2, \quad
\la_3 = |a|^2, \quad \la_4 = |b|^2 + |c|^2. \quad
\la_5 = |b|^2, \quad \la_6 = |c|^2.
$$

\item{} When employing a restricted basis set, there is no loss of
information in working with the wave function~\eqref{eq:sixpinned}.
Even in the general case, at rank six simultaneous variation of
orbitals and coefficients is still a tempting proposition for the
lithium series, in view of the following. The possible loss of
information when projecting the total wave function onto this subspace
of pinned states has been computed~\cite{supple}. Given a wave
function $|\Phi\rangle \in \wedge^3 \H_6$, and letting $P$ be the
projection operator onto the subspace spanned by the Slater
determinants $[\al_1\al_2\al_3], \,[\al_1\al_4\al_5]$ and
$[\al_2\al_4\al_6]$, we have the following upper and lower bounds for
this projection,
$$
1 - \frac{1+2\xi}{1-4\xi} \, \D \leq \|P\Phi\|^2_2 \leq 1 - \half\D,
\word{provided} \xi := 3 - \la_1 - \la_2 - \la_3 < \frac{1}{4}\,.
$$ 
Within our calculations the lower bound is larger than $99.997\%$.
Presumably, by dint of astute variation tactics one could obtain
extremely good values for the energy with just three Slater terms.

\item{} Finally, returning to the article~\cite{chinochano}, the
authors there correctly argue that their treatment of universal
subspaces gives an alternative proof for the Klyachko representability
conditions on $\wedge^3\H_6$. Conversely, the above gives an
independent proof of the assertions in~\cite{chinochano}, for the same
case.

\end{itemize}

\subsection{Reduced matrices on pinned $\wedge^3\,\H_6$}

As in the rank-five case, the one-body and one-hole matrices read
$\rho_1 = \sum^6_{i=1} \la_i\, |\al_i\rangle \langle\al_i|$ and
$\eta_1 = \sum^6_{i=1} (1-\la_i)\, |\al_i\rangle \langle\al_i|$.

The two-body and two-hole matrices are, respectively, written as
$$
\rho_2 = \sum^6_{i=1} \la_i\, |\om_i\rangle \langle\om_i|
\word{and} \eta_2 = \sum_{i=1}^6 \mu_i\, |h_i\rangle \langle h_i|
= \sum_{i\in\{1,2,4\}} \la_j\, |h_j\rangle \langle h_j|,
$$
with $|\om_i\rangle := \frac3{\sqrt{\la_i}}\langle\al_i|
\Psi\rangle$, and
\begin{align*}
|h_1\rangle &= \frac1{\sqrt{\la_1}} \bigl(b[\al_2\al_3] + a
[\al_4\al_5]\bigr), \quad |h_2\rangle = \frac1{\sqrt{\la_2}} \bigl(c
[\al_1\al_3] + a[\al_4\al_6]\bigr),
\\
|h_4\rangle &= \frac1{\sqrt{\la_4}}\bigl(c [\al_1\al_5] + b
[\al_2\al_6]\bigr);
\end{align*}
note that $|h_j\rangle=|\om_j\rangle$ for $i=3,5,6$ correspond to
single determinants. Moreover, $\eta_2^2=\eta_2$.

\subsection{$Z$-dependency of the quasi-pinning}

The dependence of the inequality \eqref{eq:pinning6++} on
the atomic number of the nucleus deserves some extra discussion. The first
occupation number will grow as the atomic charge in the nucleus
increases. Fig.~\ref{graf:Saturation} features the evolution of
the saturation when $Z$ takes values in $\{3,\ldots,12\}$. The most
relevant measure is $\D/\la_6$, which (mostly) \textit{decreases} with
$Z$. This means that the numerical distance between $\la_5$ and
$\la_4$ --- or between $\la_2$ and $\la_3$ --- is rapidly decreasing
with~$Z$.

\begin{figure}[ht] 
\centering
\begin{subfigure}[b]{.47\textwidth}
 \centering
 \includegraphics[width=\textwidth]{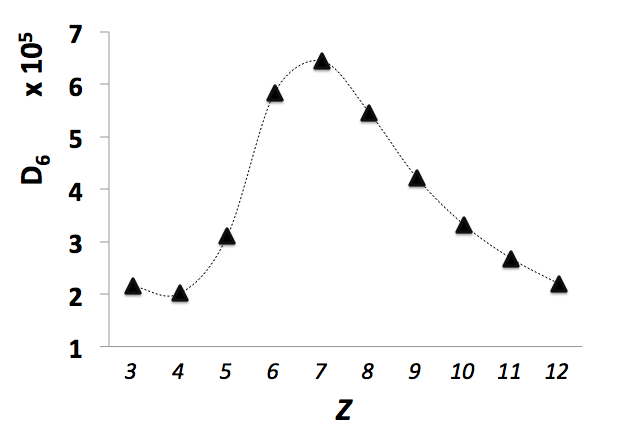}
\caption{} 
\end{subfigure}
\quad
\begin{subfigure}[b]{.47\textwidth}
\centering
\includegraphics[width=\textwidth]{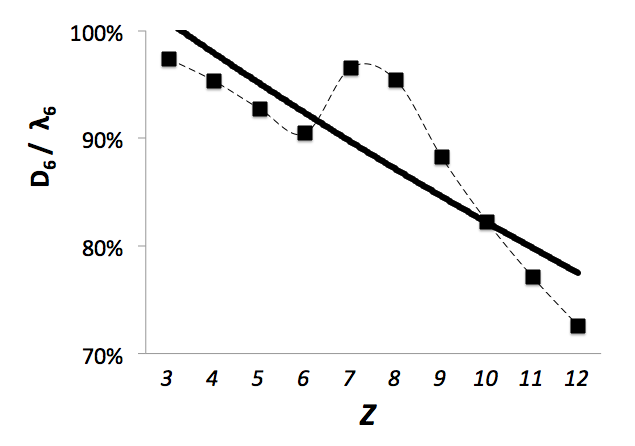}  
\caption{}
\end{subfigure}
\hfill
\caption{(a) $\D$ and (b) $\D/\la_6$ as functions of the atomic number
$Z$ for three-electron systems.}
\label{graf:Saturation}
\end{figure}

\section{Klyachko theory in higher-rank approximations}
\label{sec:seriousranks2}

\subsection{A rank-seven approximation}

We choose the basis set
$$
\begin{pmatrix} |1\rangle \\ |2\rangle \\ |3\rangle \\ |4\rangle \\
|5\rangle \\ |6\rangle \\ \ket7 \end{pmatrix} = R \begin{pmatrix}
\dl_1\up \\ \dl_1\,\dn \\ \psi^p_3\,\dn \\ \dl_2\,\dn \\ \dl_2 \up \\
\dl_3\,\dn \\ \dl_3\,\up \end{pmatrix},
$$
that is, we include also the spin-up counterpart $\ket7:= \dl_3\up$ of orbital
$\ket6$. In principle we have $35=\binom73$ Slater determinants, of
which eighteen have the total $S_z$ component equal to~$\dn$ and fifteen
are spin-adapted: the eight ones of the rank-six approximation~$6^b$ plus
$$
[267], \ [367], \ [567], \ [136] - [237], \ [346] - [357], \
[257] - [156], \ [257] - [246].
$$

There are four Klyachko inequalities for $\mathrm{Li}$ in a rank-seven
configuration:
\begin{align*}
\la_1 + \la_2 + \la_4 + \la_7 &\leq 2; \qquad \la_1 + \la_2 + \la_5 + 
\la_6 \leq 2;
\\
\la_2 + \la_3 + \la_4 + \la_5 &\leq 2; \qquad \la_1 + \la_3 + \la_4 +
\la_6 \leq 2.
\end{align*}
In our calculations we find 
\begin{align*}
0 \le \D^1_7 &= 2 - (\la_1 + \la_2 + \la_4 + \la_7) = 0, 
\\
0 \le \D^2_7 &= 2 - (\la_1 + \la_2 + \la_5 + \la_6) = 1.3045 \x 10^{-5},
\\
0 \le \D^3_7 &= 2 - (\la_2 + \la_3 + \la_4 + \la_5) = 7.7411 \x 10^{-5},
\\
0 \le \D^4_7 &= 2 - (\la_1 + \la_3 + \la_4 + \la_6) = 8.0025 \x 10^{-5}.
\end{align*}

There is a number of interesting issues on the structure of the constraints, which now shall be briefly
discussed.

\begin{itemize} 

\item{} The pinned system can be factorized,
$$
\wedge^3\H_7 \to \H_3 \ox \wedge^2\H_4.
$$
That is, the one-body matrix can be split into a direct sum of two
matrices,
$$
\rho_1 = \rho^\up_1 \oplus \rho^\dn_1.
$$
The first one is a $3 \x 3$ square matrix whose trace is equal to 1
and that is associated with the electron with spin pointing $\up$, while
the second matrix is a $4 \x 4$ square matrix whose trace is equal to two
and is associated with the two electrons with spin pointing~$\dn$.
With the numbering already dictated by the occupancies, its entries
read
\begin{align*}
\rho^\up_1(1,1) &= |c_{123} |^2 + |c_{124} |^2 + |c_{126} |^2 +
|c_{134} |^2 + |c_{146} |^2 + |c_{136} |^2
\\
\rho^\up_1(1,2) &= c_{123}c^*_{235} + c_{124}c^*_{245} -
c_{126}c^*_{256} + c_{134}c^*_{345} - c_{136}c^*_{356} -
c_{146}c^*_{456}
\\
\rho^\up_1(1,3) &= c_{123}c^*_{237} + c_{124}c^*_{247} +
c_{126}c^*_{267} + c_{134}c^*_{347} + c_{136}c^*_{367} +
c_{146}c^*_{467}
\\
\rho^\up_1(2,2) &= |c_{235} |^2 + |c_{245} |^2 + |c_{256} |^2 +
|c_{345} |^2 + |c_{356} |^2 + c_{456}
\\
\rho^\up_1(2,3) &= c_{235}c^*_{237} + c_{245}c^*_{247} -
c_{256}c^*_{267} + c_{345}c^*_{347} - c_{356}c^*_{367} -
c_{456}c^*_{467}
\\
\rho^\up_1(3,3) &= |c_{237} |^2 + |c_{247} |^2 + |c_{267} |^2 +
|c_{347} |^2 + |c_{367} |^2 + c_{467}
\end{align*} 
and 
\begin{align*}
\rho^\dn_1(1,1) &= |c_{123} |^2 + |c_{124} |^2 + |c_{126} |^2 +
|c_{235} |^2 + |c_{237} |^2 + |c_{245} |^2 + |c_{2dn47} |^2 + |c_{256}
|^2 + |c_{267} |^2
\\
\rho^\dn_1(1,2) &= c_{124}c^*_{134} + c_{126}c^*_{136} +
c_{245}c^*_{345} + c_{247}c^*_{347} + c_{256}c^*_{356} +
c_{267}c^*_{367}
\\
\rho^\dn_1(1,3) &= -c_{123}c^*_{134} + c_{126}c^*_{146} -
c_{235}c^*_{345} - c_{237}c^*_{347} + c_{256}c^*_{456} +
c_{267}c^*_{467}
\\
\rho^\dn_1(1,4) &= -c_{123}c^*_{136} - c_{124}c^*_{146} +
c_{235}c^*_{356} - c_{237}c^*_{367} + c_{245}c^*_{456} -
c_{247}c^*_{467}
\\
\rho^\dn_1(2,2) &= |c_{123} |^2 + |c_{134} |^2 + |c_{136} |^2 +
|c_{235} |^2 + |c_{237} |^2 + |c_{345} |^2 + |c_{347} |^2 + |c_{356}
|^2 + |c_{367} |^2 
\\
\rho^\dn_1(2,3) &= c_{123}c^*_{124} + c_{136}c^*_{146} +
c_{235}c^*_{245} + c_{237}c^*_{247} - c_{356}c^*_{456} +
c_{367}c^*_{467}
\\
\rho^\dn_1(2,4) &= c_{123}c^*_{126} - c_{134}c^*_{146} -
c_{235}c^*_{256} + c_{237}c^*_{267} + c_{345}c^*_{456} -
c_{347}c^*_{467}
\\
\rho^\dn_1(3,3) &= |c_{124} |^2 + |c_{134} |^2 + |c_{146} |^2 +
|c_{245} |^2 + |c_{247} |^2 + |c_{345} |^2 + |c_{347} |^2 + |c_{456}
|^2 + |c_{467} |^2
\\
\rho^\dn_1(3,4) &= c_{124}c^*_{126} + c_{134}c^*_{136} -
c_{245}c^*_{256} + c_{247}c^*_{267} - c_{345}c^*_{356} +
c_{347}c^*_{367}
\\
\rho^\dn_1(4,4) &= |c_{126} |^2 + |c_{136} |^2 + |c_{146} |^2 +
|c_{256} |^2 + |c_{267} |^2 + |c_{356} |^2 + |c_{367} |^2 + |c_{456}
|^2 + |c_{467} |^2.
\end{align*} 

\item{} For the first time we see the appearance of two \textit{scales}
of quasi-pinning.

\item{} If the second constraint were saturated, the
selection rule fixes the number of Slater determinants in the
decomposition of the wave function to be nine,
$$
[\al_1\al_2\al_3], [\al_1\al_4\al_5], [\al_1\al_4\al_6],
[\al_1\al_5\al_7], [\al_1\al_6\al_7], [\al_2\al_4\al_5],
[\al_2\al_4\al_6], [\al_2\al_5\al_7], [\al_2\al_6\al_7].
$$

\item{} As for the case of $\wedge^3\H_6$, the loss of information
when projecting the total wave function onto this nine-dimensional
subspace of twice pinned states can be estimated. In appendix B we give
a proof of the following theorem: let a wave function $|\Phi\rangle
\in \bigwedge^3\H_7$ with natural orbitals $\ket{\al_i}^7_{i=1}$,
occupation numbers $\{\la_i\}^7_{i=1}$, saturating the first
restriction. Moreover, let $P_7$ be the projection operator onto the subspace
spanned by the Slater determinants above. Then, the upper and
lower bounds of this projection are given by
\begin{align*}
1 - \frac{1 + 9\xi}{1 - 11\xi} \D^2_7 \leq \|P_7\Phi\|^2_2
\leq 1 - \half \D^2_7 \word{provided that} \xi < \frac{1}{11} \,.
\end{align*}
Within our calculations, 
$1 - \dfrac{1+9\xi}{1-11\xi}\,\D^2_7 = 1 - 1.3852 \x 10^{-5} = 99.9986\%$.

\item{} If, in addition, the third or the fourth constraint becomes
saturated, the selection rules decreases the number of allowed
determinants to just 4. Saturating both simultaneously 
reduces the case to the saturated rank-six wavefunction.

\end{itemize}

We omit the expressions of the two-body and two-hole matrices, which
can be easily calculated. It should, however, be added that the tensor
character under rotations of the reduced matrices for a three-electron
system is quite different from the one for a two-electron system; in
particular the relative weight in the lithium isoelectronic series of
the six components identified in~\cite[Sect.~6A]{Greatsmallbook}
or~\cite{Bellona} deserves some further study.

\subsection{Quasi-pinning displayed in the rank-eight approximation}

We can obtain rank eight by adding a new orbital $\ket8 := \dl_4\dn$,
giving now $\binom{3}{1}\binom{5}{2}=30$ Slater determinants
with the correct $z$-component of the spin. Among them,
$21$ are spin-adapted, i.e., the fifteen ones of the rank-seven
approximation, plus
$$
[128], \ [458], \ [678], \ [148] - [258], \ [168] - [278], \
[568] - [478].
$$

The number of Klyachko inequalities grows notably with the rank. We
find \textit{31}~inequalities in~\cite{Alturulato}. Of those, 28
constraints are displayed in our Table~\ref{table:M8}.

\begin{table}
\centering    
{\footnotesize  
\begin{tabular}{l r}  
\qquad \qquad \qquad \qquad Inequality  & Value $\x 10^{3}$  \\ [0.5ex] 
\hline\hline \\ [-0.9ex]    
$0 \le \D^{1}_8  = 2 - (\la_1 + \la_2 + \la_4 + \la_7)$ &  0.0017 \\
$0 \le \D^{2}_8  = 2 - (\la_1 + \la_2 + \la_5 + \la_6)$ &  0.0200 \\
$0 \le \D^{3}_8  = 2 - (\la_2 + \la_3 + \la_4 + \la_5)$ &  0.0671 \\
$0 \le \D^{4}_8  = 2 - (\la_1 + \la_3 + \la_4 + \la_6)$ &  0.0894 \\
\\
$0 \le \D^{5}_8  = 1 - (\la_1 + \la_2 - \la_3)$       &  0.0200 \\
$0 \le \D^{6}_8  = 1 - (\la_2 + \la_5 - \la_7)$       &  0.0854 \\
$0 \le \D^{7}_8  = 1 - (\la_1 + \la_6 - \la_7)$       &  0.1078 \\
$0 \le \D^{8}_8  = 1 - (\la_2 + \la_4 - \la_6)$       &  0.0671 \\
$0 \le \D^{9}_8  = 1 - (\la_1 + \la_4 - \la_5)$       &  0.0894 \\
$0 \le \D^{10}_8 = 1 - (\la_3 + \la_4 - \la_7)$       &  0.1548 \\
\\
$0 \le \D^{11}_8 = 1 - (\la_1 + \la_8)$               &  0.1592 \\
\\
$0 \le \D^{12}_8 = -(\la_2 - \la_3 - \la_6 - \la_7)$  &  0.0854 \\
$0 \le \D^{13}_8 = -(\la_4 - \la_5 - \la_6 - \la_7)$  &  0.1548 \\
$0 \le \D^{14}_8 = - (\la_1 - \la_3 - \la_5 - \la_7)$ &  0.1078 \\

$0 \le \D^{15}_8 = 2 - (\la_2 + \la_3 + 2\la_4 - \la_5 - \la_7 + \la_8)$ &  1.4183 \\
$0 \le \D^{16}_8 = 2 - (\la_1 + \la_3 + 2\la_4 - \la_5 - \la_6 + \la_8)$ &  0.2956 \\
$0 \le \D^{17}_8 = 2 - (\la_1 + 2\la_2 - \la_3 + \la_4 - \la_5 + \la_8)$ &  1.2836 \\
$0 \le \D^{18}_8 = 2 - (\la_1 + 2\la_2 - \la_3 + \la_5 - \la_6 + \la_8)$ &  0.1569 \\
\\
$0 \le \D^{19}_8 = - (\la_1+ \la_2 - 2\la_3 - \la_4 - \la_5)$  &  1.2897 \\
\\
$0 \le \D^{21}_8 = - (\la_1 - \la_3 - \la_4 - \la_5 + \la_8)$  &  1.4288 \\
\\
$0 \le \D^{23}_8 = 1 - (2\la_1 - \la_2 + \la_4 - 2\la_5 - \la_6 + \la_8)$  & 0.3894 \\
$0 \le \D^{24}_8 = 1 - (\la_3 + 2 \la_4 - 2\la_5 - \la_6 - \la_7 + \la_8)$ & 1.5591 \\
$0 \le \D^{25}_8 = 1 - (2\la_1 - \la_2 - \la_4 + \la_6 - 2\la_7 + \la_8)$  & 0.4262 \\
$0 \le \D^{26}_8 = 1 - (2\la_1 + \la_2 - 2 \la_3 - \la_4 - \la_6 + \la_8)$ & 0.2507 \\
$0 \le \D^{27}_8 = 1 - (\la_1 + 2\la_2 - 2 \la_3 - \la_5 - \la_6 + \la_8)$ & 1.3551 \\
\\
$0 \le \D^{29}_8 = \la_1 - \la_3 - 2\la_4 + 3\la_5 + 2\la_6 + \la_7 - \la_8$ & 2.8758 \\
$0 \le \D^{30}_8 = -(2\la_1 + \la_2 - 3\la_3 - 2\la_4 - \la_5 - \la_6 + \la_8)$ & 1.5204 \\
$0 \le \D^{31}_8 = -(\la_1 + 2\la_2 - 3\la_3 - \la_4 - 2\la_5 - \la_6 + \la_8)$ & 2.6247 \\
[1ex]   
\hline    
\end{tabular} }
\caption{Klyachko inequalities for a system $\wedge^3\H_8$ and some
numerical values for $\mathrm{Li}$.}
\label{table:M8} 
\end{table}

In the table, we have included the values of the inequalities that result from our calculation,
and in order to analyze those further, we have plotted them 
both in a linear and in a logarithmic scale in Fig.~\ref{graf:Plot}.
The presence of several scales is clearly seen. Moreover, conditions involving
the eighth occupation number are clearly weaker than the previous
ones. The main point, which both confirms and extends the
findings for the toy model in~\cite{pinning}, is the
\textbf{\textit{robustness}} of quasipinning. In particular, the
quantity $\D^{1}_8$, found to be exactly zero in the previous rank,
remains in a strongly pinned regime.

\begin{figure}[ht] 
\centering
\begin{subfigure}[b]{.47\textwidth}
\centering
\includegraphics[width=\textwidth]{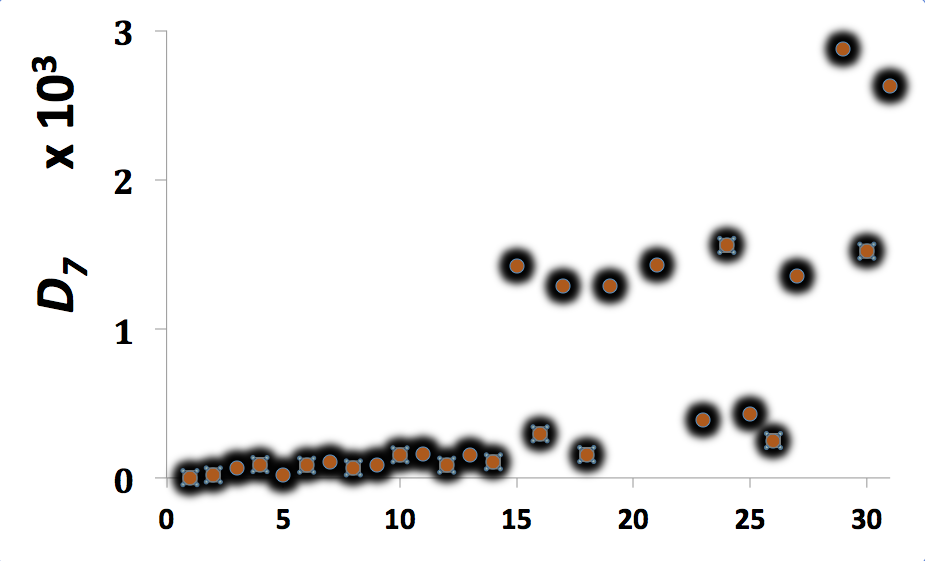}
\caption{} 
\end{subfigure}
\quad
\begin{subfigure}[b]{.47\textwidth}
\centering
\includegraphics[width=\textwidth]{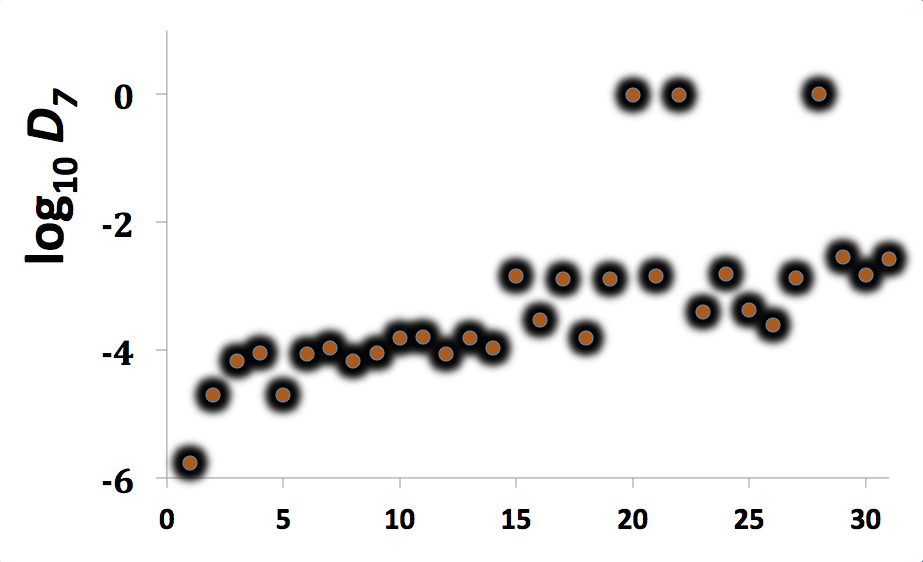}  
\caption{}
\end{subfigure}
\hfill \caption{Plot of the inequalities of Table \ref{table:M8} in
(a) values of $10^{-3}$ and (b) a logarithmic scale.}
\label{graf:Plot}
\end{figure}

Finally, one can examine the effect of the saturation conditions and
the resulting dramatic reduction of the number of Slater determinants,
as well as the simplification of the corresponding two-body matrix.
This remarkable evolution is visualized in Fig.~\ref{graf:satur}.

\begin{figure}
\centering
\includegraphics[width=8cm]{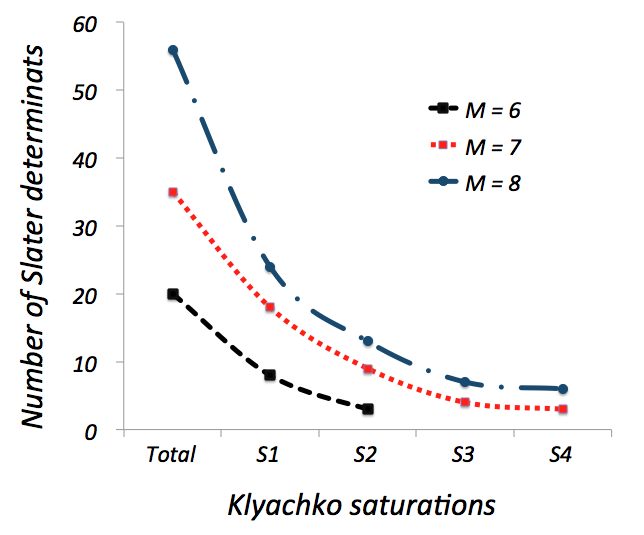}
\caption{(Color online) Evolution of the dimension of the space of Slater
determinants as a function of Klyachko's saturations for rank six,
seven and eight.}
\label{graf:satur}
\end{figure}

\section{Conclusion}
\label{sec:conclusion}

By means of numerical calculations, we have explored the nature of the
quasi-pinning in real three-electron atoms. In the space
$\wedge^3\H_6$, for \textit{restricted} spin-orbitals we find that the
Klyachko constraint is saturated. For \textit{unrestricted}
configurations, quasi-pinning is bounded by the lowest occupation
number~$\la_6$. In approximations of larger rank, the Klyachko
constraints split into well differentiated groups of different levels
of saturation. In other words, for a real system we find results
compatible with those found previously for the model system
\cite{pinning}. A simple geometric probability argument also suggests
stability of quasi-pinning. Moreover, whenever $3-\la_1-\la_2-\la_3$
is not far from zero, projecting the complete wave function of the CI
picture into the set of pinned states appears to result in negligible
loss of information. Thus, the Klyachko-guided addition of a few
Slater determinants to Hartree--Fock type states becomes a low-cost
approach to accurate atomic wave functions.
 
In addition, through our work we have verified some recent results of
quantum information theory for three-electron systems. In particular,
we find non-trivial, but fairly low, quantum entanglement in ground
states.

Even if the present study is a step towards a general description of
real systems, the system of our study is special in a couple of
aspects. The Hartree--Fock approximation is a good starting point for
the $\mathrm{Li}$-like atoms, and the system has a high spherical
symmetry. In future work we hope to address cases where the occupation
numbers will not lie so close to 0 or 1, and/or the symmetry is
reduced. We will thereby be able to explore whether the conclusions we
have drawn in the present study hold more generally in chemistry.

\appendix

\section{On the nature of the Klyachko restrictions}
\label{sec:Klyachko}

Here, we shall not give proofs of the Klyachko
constraints but just discuss some few aspects of relevance to the present
work. It is useful to consider the skew Cauchy
formula,
\begin{equation}
\bigwedge\nolimits^{\!N} \bigl(\H_\mathrm{s} \ox \H_\mathrm{orb}\bigr)
= \sum_{|\kappa|=N}
\H_\mathrm{s}^\kappa \ox \H_\mathrm{orb}^{\bar\kappa}.
\label{eq:cirque-du-soleil}
\end{equation}
Here, $\kappa$ denotes the representation corresponding to the
partition or Young tableau~$\kappa$, and $\bar\kappa$ is the dual
partition.

In the present work we have exclusively the case $\dim\H_\mathrm{s}=2$ and $N=3$, 
which makes everything relatively simple. The three-electron
state space splits into spin-orbital sectors, which one needs to
specify in order to check quasi-pinning, as well as to gauge
entanglement. $\H_\mathrm{s}$ corresponds to a spin-$\half$
particle. Therefore, on the left hand side we may have only
representations of $SU(2)$, i.e., either $j = \half$ or
$j = \tfrac{3}{2}$ for three particles. Since there are no skewsymmetric
combinations of three spins-$\half$, the partition $(1,1,1)$ on the
right hand side plays no role; consequently, only tableaux with up to
two columns may appear on the left hand side.

Consider for instance the first non-trivial case $\wedge^3(\H_2\ox
\H_3)$ in the configuration $6^a$ of Section~\ref{sec:seriousranks1}.
There are 20 configurations in all. Clearly there is one with three
spin down and one with three spin up, belonging to the representation
with $j = \tfrac{3}{2}$. Of the eighteen remaining states, nine have
one spin down in total, and nine have spin up. But only eight of
each belong to the $j = \half$ representation; the other two belong to
$j = \tfrac{3}{2}$; whereby the spatial orbitals enter in the unique
completely skewsymmetric combination. This takes care of ``spin
contamination''. Accordingly,
$$
\wedge^3\bigl(\H_{2\,\mathrm s} \ox \H_{3\,\mathrm{orb}}\bigr) =
\dn\H_2^{\ox3} \oplus \up \H_2^{\ox3} + \H_\mathrm{s}^{3/2}
\ox \wedge^3\H_{3\,\mathrm{orb}}.
$$
From these simple observations to the generalized Pauli constraints
there is still a long haul, demanding generous dollops of Kirillov's
theory of orbits of the coadjoint action for compact
groups~\cite{Alturulato}; the surprising outcome is that only linear
inequalities are found.

Of course, not all of our basis sets conform to the left hand side
of~\eqref{eq:cirque-du-soleil}. This causes no problem, however, since
any basis set can be considered a special case of a larger one
with the ``right'' structure, with some holes. Important is it that the
Klyachko restrictions are \textit{consistent}, so lower-rank ones can
be derived from higher-rank ones. Recall for instance our example
$\wedge^3\H_7$, and the four corresponding relations:
\begin{align*}
\la_1 + \la_2 + \la_4 + \la_7 &\leq 2; \qquad \la_1 + \la_2 + \la_5 + 
\la_6 \leq 2;
\\
\la_2 + \la_3 + \la_4 + \la_5 &\leq 2; \qquad \la_1 + \la_3 + \la_4 +
\la_6 \leq 2.
\end{align*}
At first, the original Pauli principle $\la_1\le1$ is perhaps not
entirely obvious here; it follows from summing the second and the
fourth. Also, let us consider the case $\la_7=0$. Then summing the
second and the third we obtain $\la_2+\la_5\le1$; the second and
fourth yield $\la_3+\la_4\le1$, and so on: we plainly recover the
Borland--Dennis relations for~$\wedge^3\H_6$. The reader will have no
difficulty in retrieving all the lower-rank relations from the ones
on~$\wedge^3\H_8$.

\section{Bounds for the rank-seven approximation }
\label{sec:bounds}

\begin{thm}
Let $|\Phi\rangle \in \wedge^3 \H_7$ be a wave function with natural
orbitals $\{|\al_i\rangle\}^7_{i=1}$ and occupation numbers
$\{\la_i\}^7_{i=1}$ arranged in decreasing order. Let $P_7$ be the
projection operator onto the subspace spanned by the Slater
determinants
$$
[\al_1\al_2\al_3], \ [\al_1\al_4\al_5], \ [\al_1\al_4\al_6], \
[\al_1\al_5\al_7], \ [\al_1\al_6\al_7], \ [\al_2\al_4\al_5], \
[\al_2\al_4\al_6], \ [\al_2\al_5\al_7], \ [\al_2\al_6\al_7].
$$
Upper and lower bounds of this projection are given by
\begin{align}
1 - \frac{1+9\xi}{1-11\xi}\, \D^2_7 \leq \|P_7\Phi\|^2_2
\leq 1 - \half \D^2_7 ,
\word{where} \xi := 3 - \la_1 - \la_2 - \la_3 < \frac{1}{11} \,.
\label{eq:ulbounds}
\end{align}
\end{thm}
\begin{proof}
Let $J_1 = \{3,5,6\}$ and $J_2 = \{1,2,4,7\}$. A general wave function
in $ \wedge^3\H_7$ is given by
$$
|\Phi\rangle = \sum_{\substack{k \in J_1\\ i,j \in J_2}} c_{ijk}
[\al_i\al_j\al_k] \word{and consequently} \la_m = \sum_{\substack{k
\in J_1\\ i,j \in J_2 \\ m \in\{i,j,k\} }} |c_{ijk} |^2.
$$
Therefore,
\begin{align*}
\la_1 &= |c_{123} |^2 + |c_{125} |^2 + |c_{126} |^2 + |c_{134} |^2 +
|c_{145} |^2 + |c_{146} |^2 + |c_{137} |^2 + |c_{157} |^2 + |c_{167}|^2,
\\
\la_2 &= |c_{123} |^2 + |c_{125} |^2 + |c_{126} |^2 + |c_{234} |^2 +
|c_{245} |^2 + |c_{246} |^2 + |c_{237} |^2 + |c_{257} |^2 + |c_{267}
|^2,
\\
\la_5 &= |c_{125} |^2 + |c_{145} |^2 + |c_{157} |^2 + |c_{245} |^2 +
|c_{257} |^2 + |c_{457} |^2,
\\
\la_6 &= |c_{126} |^2 + |c_{146} |^2 + |c_{167} |^2 + |c_{246} |^2 +
|c_{267} |^2 + |c_{467} |^2 .
\end{align*}
A simple computation gives
\begin{align*}
&\la_1 +\la_2 +\la_5 +\la_6 = 2 |c_{123} |^2 + 3 |c_{125} |^2 + 3
|c_{126} |^2 + |c_{134} |^2 + 2 |c_{145} |^2 + 2 |c_{146} |^2 +
|c_{137} |^2
\\
&\; + 2 |c_{157} |^2 + 2 |c_{167} |^2 + |c_{234} |^2 + 2 |c_{245} |^2
+ 2|c_{246} |^2 + |c_{237} |^2 + 2 |c_{257} |^2 + 2 |c_{267} |^2 +
|c_{457} |^2
\\
&\; + |c_{467} |^2. \word{And trivially} \D^2_7 = 2 - (\la_1 +\la_2
+\la_5 +\la_6) = L + S - |c_{125} |^2 - |c_{126} |^2,
\end{align*}
where $S := 2 |c_{347} |^2 + |c_{457} |^2 + |c_{467} |^2$ and $L :=
|c_{134} |^2 + |c_{137} |^2 + |c_{234} |^2 + |c_{237} |^2$. Thus, for
the projection onto the aforementioned subspace we have
\begin{align*}
\| P_7 \Phi\|^2_2 &= |c_{123}|^2 + |c_{145}|^2 + |c_{146}|^2 +
|c_{157}|^2 + |c_{167}|^2 + |c_{245}|^2 + |c_{246}|^2 + |c_{257}|^2 +
|c_{267}|^2
\\
&= 1 - (L + |c_{347} |^2 + |c_{125} |^2 + |c_{457} |^2 + |c_{126} |^2
+ |c_{467} |^2 )
\\
&\leq 1 - \half (L + 2 |c_{347} |^2 - |c_{125} |^2 + |c_{457} |^2 -
|c_{126} |^2 + |c_{467} |^2 ) = 1 - \half \D^2_7,
\end{align*}
which is the upper bound of \eqref{eq:ulbounds}.

To establish the lower bound, note that in the basis of natural
orbitals we know that
\begin{align*}
\langle \al_6 | \rho_1 | \al_3 \rangle &= c_{123}^* c_{126} -
c_{134}^* c_{146} + c_{137}^* c_{167} -c_{234}^* c_{246} + c_{237}^*
c_{267} - c_{347}^* c_{467} = 0
\\
\langle \al_5 | \rho_1 | \al_3 \rangle &= c_{123}^* c_{125} -
c_{134}^* c_{145} + c_{137}^* c_{157} +c_{234}^* c_{245} + c_{237}^*
c_{257} - c_{347}^* c_{457} = 0.
\end{align*}
Let $\epsilon := c_{123}$, the amplitude of the Hartree--Fock
determinant. Using the Cauchy inequality $(A + B+ C+ D+ E)^2 \leq
5(A^2 + B^2 + C^2 + D^2 + E^2)$ as well as $|c_{abc} |^2 \leq 1 -
|\epsilon |^2$ whenever $abc \neq 123$, we obtain:
\begin{align*}
|c_{126} |^2 &\leq \frac{5}{|\epsilon|^2} \big[|c_{134} |^2 |c_{146}
|^2 + |c_{137} |^2 |c_{167} |^2 + |c_{234} |^2 |c_{246} |^2 + |c_{237}
|^2 |c_{267} |^2 + |c_{347} |^2 |c_{467} |^2\big]
\\
&\leq \frac{5(1- |\epsilon |^2)}{|\epsilon |^2} \Big[L + \half (
|c_{347} |^2 + |c_{467} |^2) \Big] \word{and}
\\
|c_{125} |^2 &\leq \frac{5}{|\epsilon |^2} \big[|c_{134} |^2 |c_{145}
|^2 + |c_{137} |^2 |c_{157} |^2 + |c_{234} |^2 |c_{245} |^2 + |c_{237}
|^2 |c_{257} |^2 + |c_{347} |^2 |c_{457} |^2\big]
\\
&\leq \frac{5(1- |\epsilon |^2)}{|\epsilon |^2} \Big[L + \half (
|c_{347} |^2 + |c_{457} |^2) \Big].
\end{align*}
Let us set, for some $r,u\geq0$,
\begin{align*}
&L + |c_{347} |^2 + |c_{125} |^2 + |c_{457} |^2 + |c_{126} |^2 +
|c_{467} |^2
\\
&\leq L + (1 + u) S + (1 - r) (|c_{125} |^2 + |c_{126} |^2) + r
(|c_{125} |^2 + |c_{126} |^2)
\\
& \leq L + (1 + u) S + (1 - r) (|c_{125} |^2 + |c_{126} |^2) +
\frac{5r(1- |\epsilon |^2)}{|\epsilon |^2} \Big[2L + \half S \Big] \\
&= \Bigg[1+\frac{10r(1- |\epsilon |^2)}{|\epsilon |^2} \Bigg] L + (1 -
r) (|c_{125} |^2 + |c_{126} |^2) + \Bigg[ (1 + u) + \frac{5r(1-
|\epsilon |^2)}{2|\epsilon |^2} \Bigg] S .
\end{align*}
By choosing
\begin{align*}
r = \frac{2|\epsilon |^2}{11|\epsilon |^2-10} \word{and} u =
\frac{15(1- |\epsilon |^2)}{11|\epsilon |^2-10} \word{with} |\epsilon
|^2 > 10/11,
\end{align*} 
we obtain that $L + |c_{347} |^2 + |c_{125} |^2 + |c_{457} |^2 +
|c_{126} |^2 + |c_{467} |^2 \leq (r-1) \D_7^2 $. It is now clear that
\begin{align*}
\|P_7 \Phi\|^2_2 &= 1 - (L + |c_{347} |^2 + |c_{125} |^2 + |c_{457}
|^2 + |c_{126} |^2 + |c_{467} |^2 )
\\
& \geq 1 - (r - 1) \D_7^2 = 1 - \frac{1 + 9(1 - |\epsilon |^2 )}{1 -
11 (1 - |\epsilon |^2 )} \D_7^2 \geq 1 - \frac{1 + 9\xi}{1 - 11 \xi}
\D_7^2,
\end{align*}
where in the last inequality we have used $1 - |\epsilon |^2 \leq \xi
= 3 - \la_1 - \la_2 - \la_3$, which is Lemma~3 in~\cite{supple}.
\end{proof}

\section*{Acknowledgments}

CLBR was supported by a Francisco Jos\' e de Caldas scholarship,
funded by Colciencias. He very much appreciates the warm atmosphere of
the Physikalische und Theoretische Chemie group at Saarlandes
Universit\"at. JMGB was, hopefully, supported by the Diputaci\'on
General de Arag\'on (grant DGIID-DGA-E24/2) and the defunct Ministerio
de Educaci\'on y Ciencia (grant FPA2009-09638). We are grateful to
J.~C.~V\'arilly for helpful comments.

\end{document}